\documentclass[journal]{IEEEtran}

\usepackage{cite}
\usepackage{algorithm}
\usepackage[noend]{algpseudocode}
\usepackage{subcaption}
\makeatletter
\def\BState{\State\hskip-\ALG@thistlm}
\makeatother
\usepackage{enumitem} 
\usepackage{tablefootnote}
\usepackage{caption}
\usepackage{graphicx}
\usepackage{booktabs}
\usepackage{amsfonts}
\usepackage{amsthm}
\usepackage{amsmath,bm}
\usepackage{amssymb}
\usepackage{mathtools}
\usepackage{multirow}
\usepackage{hyperref}
\usepackage{bbm}

\newtheorem{exmp}{Example}
\newtheorem{theo}{Theorem}

\newtheorem{rem}{Remark}
\newtheorem{lemma}{Lemma}

\newtheorem{definition}{Definition}[section]

\usepackage{array}
\usepackage{amsmath}

\begin{document}

\title{A Price-Based Iterative Double Auction for Charger Sharing Markets}

\author{Jie Gao,%~\IEEEmembership{Member,~IEEE,}
        ~Terrence Wong,%~\IEEEmembership{Fellow,~OSA,}
        ~Chun Wang,~\IEEEmembership{Member,~IEEE,}
        and~Jia Yuan Yu,~\IEEEmembership{Member,~IEEE}
\thanks{The authors are with the Concordia Institute for Information Systems Engineering (CIISE), Concordia University, Montréal, QC H3G 1M8,
Canada (e-mail: jie.gao@mail.concordia.ca;
te\_ong@encs.concordia.ca;
chun.wang@concordia.ca;
jiayuan.yu@concordia.ca).}
}

\maketitle

\begin{abstract}

The unprecedented growth of demand for charging electric vehicles (EVs) calls for novel expansion solutions to today's charging networks. Riding on the wave of the proliferation of sharing economy, Airbnb-like charger sharing markets
opens the opportunity to expand the existing charging networks without requiring costly and time-consuming infrastructure investments, yet the successful design of such markets relies on innovations at the interface between game theory, mechanism design, and large scale optimization. 
In this paper, we propose a price-based iterative double auction for charger sharing markets where charger owners rent out their under-utilized chargers to the charge-needing EV drivers. Charger owners and EV drivers form a two-sided market which is cleared by a price-based double auction. Chargers'  locations, availability, and time unit costs as well as the EV drivers' time, distance constraints, and preferences are considered in the allocation and scheduling process. The goal is to compute social welfare maximizing allocations 
which benefits both charger owners and EV drivers and, in turn, ensure the continuous growth of the market. We prove that the proposed double auction 
is budget balanced, individually rational, and that it is a weakly dominant strategy for EV drivers and charger owners to truthfully report their charging time constraints. In addition, results from our computation study show that the double auction achieves on average 94\% efficiency compared with the optimal solutions and scales well to larger problem instances.  
\end{abstract}


\IEEEpeerreviewmaketitle

\section{Introduction}
\IEEEPARstart{E}{nergy} spent on charging electric vehicles (EVs) will grow tremendously in the next decade. As estimated by the International Energy Agency, annual charging energy demand for the EV population is projected to increase from 58 billion kilowatt-hours to 640 billion kilowatt-hours from 2020 to 2030.
This surging demand places an unprecedented strain on existing charging networks which need to be substantially expanded in terms of the total amount of energy delivered and their geographical coverage.
%To accommodate the ever-growing energy demand of electric vehicles, existing EV charging networks have to be expanded significantly in terms of the total amount of energy delivered and their geographical coverage.
However, traditional methods to expand the charging networks such as building new charging stations and upgrading to high speed DC chargers are often costly and time-consuming. In recent years, \textit{charger sharing} has emerged as one of the cost-effective and immediate solutions to expand the existing charging networks~\cite{plenter2018assessment,vanrykel2018fostering}.
%which allows the private charger owners to share their chargers to EV drivers, has emerged as a promising alternatives to address this gap. In recent years, 
Online charger sharing platforms are being built to connect private charger owners and EV drivers. Some popular ones include
PlugShare\footnote{https://www.plugshare.com/}, EVMatch\footnote{https://www.evmatch.com/}, ChargeHub\footnote{https://chargehub.com/en/}, Share\&Charge\footnote{https://shareandcharge.com/}, CHRG Network\footnote{https://chrg.network/} and ELbnb\footnote{https://www.thelocal.se/tag/elbnb}. 

Using such platforms, private charger owners aim to rent out under-utilized chargers to recoup their installation and operation costs and EV drivers wish to procure the charging services to satisfy their energy needs.

The success of charger sharing platforms hinges on two major issues: (i) attracting both charger owners and EV drivers to the platform by providing added values to both of the groups, and (ii) retaining them by computing charging scheduling and pricing solutions which maximize the social welfare across all participants. 
Social welfare maximization benefits both EV drivers and charger owners and ensures sustainable growth of these platforms. 
The first issue requires the engineering of individual rationality~\cite{mas1995microeconomic} into the sharing mechanisms, while the second issue calls for market-based charger sharing mechanisms which optimize the overall resource allocation in game-theoretic settings.

At the present time, the main scheduling mechanisms used by the charger sharing platforms are variants of the First-Come-First-Served (FCFS) mechanism with the ``take-it-or-leave-it'' pricing schemes. While these mechanisms do motivate charger owners and EV drivers to participate if the price is right, however, they do not possess the property of individual rationality. Furthermore, these mechanisms do not explicitly consider game-theoretic behaviors of participants in their mechanism design, nor ensure the optimality of computed scheduling solutions even when FCFSs are replaced with a centralized optimization algorithm due to the market nature of the charger sharing environment. In market environments, EV drivers and charger owners are independent and self-interested agents, and the optimization algorithm may not have access to all the needed charging scheduling information because the information is privately held by the agents who cannot be assumed to follow the algorithm but rather their own self-interests. 
In this decentralized setting, agents may behave strategically in the
pursuit of their own benefits rather than the system-wide optimality. 

Charger sharing mechanism design is a relatively new research topic and the literature on it is limited. As a more general research area, EV charging scheduling has attracted increased attention in the past years. Comprehensive reviews can be found in~\cite{rahman2016review} and~\cite{wang2016smart}. Several studies~\cite{qin2011charging,de2014optimal,kang2015centralized,kang2017optimal,darabi2016intelligent} have tackled EV charging scheduling problems by applying centralized approaches, which assume that a central scheduler is responsible for all allocation decisions. However, 
these centralized approaches cannot be applied to EV charging scheduling problems in the context of a charger sharing market. This market is naturally decentralized~\cite{mas1995microeconomic, wellman2001auction}, in which scheduling information are distributed and controlled by different self-interested agents. 

Market based approaches, such as auctions, have gained popularity in providing socially desirable solutions to decentralized EV charging scheduling problems. These approaches respect autonomy and private information inherited from a distributed system and can provide incentives for agents to reveal truthful information~\cite{parkes2001auction}. 
For example, P. Samadi \textit{et al.}~\cite{samadi2011optimal} propose a Vickrey-Clarke-Groves (VCG)
based mechanism for EV charging scheduling with the objective of maximizing the social welfare. In a related approach, J.de Hoog \textit{et al.}~\cite{de2015market} design a market mechanism for smart charging that optimally allocates available capacity and, at the same time, ensures network stability. 

However, these studies address the setting of one-sided markets with one charger supplier. They cannot be directly applied to two-sided charger sharing markets. 
In~\cite{gerding2013two}, an EV charging scheduling problem is studied in a two-sided market. The authors propose a VCG payment rule to ensure truthfulness of EV drivers and charging stations. Although the VCG mechanism is well known for being truthful and socially optimal, implementations of VCG-type mechanisms generally suffer from excessively high computational costs~\cite{ausubel2006lovely} and are impractical for charger sharing markets with large numbers of charger owners and EV drivers.

Two-sided markets which involve two distinct groups of players, e.g., stock markets, are normally cleared by double auctions.
In his seminal paper,  McAfee~\cite{mcafee1992dominant} proposes a trading reduction rule to achieve truthfulness in two-sided markets with homogeneous single unit goods. For the same problem, Chu and Shen~\cite{chu2008truthful} propose an agent competition mechanism by applying shadow prices to achieve strategy proof. %However, these approaches are 
More recently, some research has attempted to design double auctions for multi-unit heterogeneous trading problems. For example,
Y. Chen \textit{et al.}~\cite{chen2013tames} extend McAfee's mechanism to multi-unit heterogeneous settings. They apply the proposed mechanism to spectrum allocation problems.  In~\cite{chichin2016towards}, a two-sided combinatorial greedy allocation mechanism is applied to multi-unit heterogeneous  cloud exchange markets. In addition, 
iterative double auctions based on a decomposition scheme have been proposed for multi-unit heterogeneous energy trading environments~\cite{majumder2014efficient, faqiry2018double, kang2017enabling,iosifidis2015double, deng2014distributed}. In these double auctions, trading goods are distinct, indivisible items. In order to apply these double auctions to the charger sharing scheduling problem, the continuous scheduling time line has to be discretized, such that the charging time period can be converted to a set of distinct time units~\cite{wellman2001auction,kutanoglu1999combinatorial}.
However, to maintain time accuracy, the discretized time unit cannot be too large. Therefore,
this approach can generate a large number of distinct time units, which
inflicts heavy computation burden on double auctions in terms of bids evaluation, communication, and winner determination~\cite{wang2009constraint, wang2012requirement}.
In~\cite{wang2011due} and~\cite{wang2013service}, the authors use scheduling specific bidding language for decentralized scheduling problems, which models scheduling related constraints naturally and reduces computation costs. %Specifically, C. Wang et 
However, both papers focus only on one-sided settings where there is only one seller in the market.

In this paper, we design a price-based iterative double auction for the charger sharing scheduling (CSS) problem in two-sided sharing markets. In this auction,
charger owners (\textit{sellers}) submit asks to indicate their available charging times, locations and time unit costs. EV drivers (\textit{buyers}) place bids to express their charging time requirements and prices they are willing to pay. 
The
auction then allocates buyers to sellers through iterative bidding with the aim of maximizing the
social welfare which is the difference between EV drivers'
total values and charger owners' total costs. We assume that sellers and buyers follow myopic best-response strategies~\cite{parkes2006iterative, parkes2001iterative}, which means they place best-response bids to the current prices. This is a reasonable assumption for large markets with many participants and for bidders with bounded computation capacities and/or limited information~\cite{iosifidis2015double}. The main contributions are summarized as follows.
\begin{itemize}
    \item  We formulate a new CSS problem with the objective of maximizing the social welfare of both EV drivers and charger owners in a two-sided charger sharing market.
    \item We propose a price-based iterative double auction which guarantees individual rationality, budget balance and allocative efficiency. In addition, the designed mechanism ensures that reporting truthful charging time constraints is a weakly dominant strategy for EVs and charger owners. This mechanism only incurs a small communication overhead and clears the market without knowing the private information of the participants.
    \item We design a fast winner determination algorithm based on simulated annealing meta-heuristic. Equipped with this algorithm, the proposed double auction can compute high-quality solutions to large scale CSS problem with low computational cost.
  %  \item The proposed mechanism considers the realistic properties of charger sharing problem. 
\end{itemize} %

The rest of the paper is organized as follows. In Section~\ref{pro}, we describe the charger sharing scheduling problem and present its mathematical model. In Section~\ref{IDA}, we present the structure and components of the
proposed auction. A theoretical analysis
on the properties of the auction is provided in Section~\ref{tr} followed by a computational study in Section~\ref{na}. Finally,
in Section~\ref{c}, we conclude our work and discuss future improvements.
\section{The charger sharing scheduling problem}
\label{pro}
We consider a charger sharing market which consists of a set of private charger owners, a set of EV drivers (henceforth called sellers and buyers), and a scheduler. Each seller provides his or her charging service offer to the market. The offer of a seller includes a location, an available time window and a cost of one time unit charging in the available time window. Buyers have preferences over charging services offered by different sellers. Each buyer prefers a seller with the most convenient location, charging time and the lowest service cost. Those preferences are quantified by the highest prices they want to pay, which will be referred to as \textit{values}.
%are interested in receiving a service from one of the available sellers.
%over a finite discrete time interval $\mathcal{T} = \{1, 2, \dots,T\}$ in which $\mathcal{T}$ is divided into $T$ equal length slots, e.g., half-hour and buyers are interested in receiving a service from one of the available sellers. 
The decision of the scheduler is to allocate each buyer to a seller, such that the social welfare (the difference between the allocated buyers' values and the allocated sellers' costs) is maximized and the scheduling constraints of both buyers and sellers are satisfied. 

Let $\mathcal{M}$ be the set of sellers and let $\mathcal{N}$ be the set of buyers, where $m\in \mathcal{M}$ defines an arbitrary seller and $n\in\mathcal{N}$ an arbitrary buyer. Each seller $m$ submits his or her location and \textit{type} to the charger sharing market. Let $\boldsymbol{l}_m = (f_m, g_m)$ be the location of seller $m$, where $f_m$ and $g_m$ represent the latitude and the longitude of the location, respectively.
A type of each seller $m$ is defined as follows:
\begin{definition}
Each seller $m$ is characterized by a ``type" $\alpha_m =
(s_m, e_m, c_m)$, which denotes the service start time, service end time and the charging cost for one time unit, respectively.
\label{def1}
\end{definition}
%Let $\boldsymbol{l}_m = (f_m, g_m)$ be the location of seller $m$, where $f_m$ and $g_m$ represent the latitude and the longitude respectively. This is a public information.
%The offered charging service, also called the \textit{ask}, of each seller $m$ is characterized by a vector $\alpha_m = (s_m, e_m, c_m)$, where 
We refer to $[s_m, e_m]$ as the available charging time window offered by seller $m$, during which the seller is available for providing its charger to EV drivers. More specifically, $s_m$ and $e_m$ are the earliest and latest service times of seller $m$. 
$c_m \in \mathbb{R}^+$ is the charging cost of seller $m$ for a time unit.
%for selling a time unit (e.g., half-hourly unit). 
The cost consists of \textit{electricity cost} and \textit{parking cost}. Electricity cost is the cost of electricity consumed within a time unit for charging a vehicle. Parking cost refers to the one time unit cost associated with the operation of the charging space, property tax, insurance, maintenance or rental costs.

In this market, a buyer may have different \textit{types} for sellers due to the different travelling distances to those sellers.
A type of buyer $n$ is defined as follows:
%In this paper, electricity unit cost is fixed, while location cost (includes a proportion of the property tax and also the charger maintenance cost) is differ for each seller. %This location cost is the cost the charger owner incurred to offer the charging space which includes a proportion of the property tax and also the charger maintenance cost.
%Specifically,  location cost is estimated by sellers based on their house prices. Demand cost is calculated by the probability of charging demand, multiplied by the time unit cost. Typically, peak hours are expected to have a higher demand probability, and so have a higher demand cost. 
\begin{definition}
Each buyer $n$ is characterized by a ``type"  $\theta_{n,m}= (a_{n,m}, d_{n,m}, r_{n,m}, v_{n,m})$ for a particular seller $m$, indicating the arriving time, departure time, required charging duration and the charging value or willingness to pay, respectively. The set of all possible types of buyer $n$ is denoted by $\boldsymbol \theta_n$.
%The \textit{bid} of buyer $n$ is represented by a vector$\boldsymbol \theta_n = (\theta_{n,m}: m\in \mathcal{M})$, where each element  $\theta_{n,m}= ( a_{n,m}, d_{n,m}, r_{n,m}, v_{n,m})$ refers to the buyer's bid for a seller $m\in \mathcal{M}$.
\label{def2}
\end{definition}
We refer to $[a_{n,m}, d_{n,m}]$ as the available charging time window of buyer $n$ for seller $m$. $a_{n,m}$ and $d_{n,m}$ indicate the earliest time that buyer $n$ can start to charge at seller $m$ and the latest time by which this buyer has to finish at seller $m$, respectively.  Note that, for buyer $n$, if $a_{n,m} < s_{m}$, this buyer has to wait until the available charging time window of seller $m$ is open. This means the earliest possible start charging time of buyer $n$ at seller $m$ is $s_m$. On the other hand, if $d_{n,m} > s_m$ for a seller $m$, buyer $n$ has to finish at time $s_m$, even his or her latest departure time is $d_{n,m}$.  
%For example, in Fig.~\ref{example}, the earliest start charging time of buyer 1 at seller 1 is 13:00. While, the latest charging finish time of buyer 1 at seller 2 is 19:00.
$r_{n,m}\in \mathbb{R}^+$ is the required charging duration of buyer $n$ for seller $m$. 
A charging request is non-preemptive, i.e., once buyer $n$ is started charging at seller $m$, it must continue charging for $r_{n,m}$ time units.
Each buyer $n$ can have different values regarding the sellers he or she would like to charge at. 
The valuation $v_{n,m}\in \mathbb{R}^+$ indicates the maximum price for which  buyer $n$ is willing to pay for charging at seller $m$.

We assume that buyers are \textit{single-minded}~\cite{lehmann2002truth}. A single-minded buyer only has two possible states: either he or she obtains the entire charging duration $r_{n,m}$ for a particular seller $m$ within his or her available charging time window and derives a value of $v_{n,m}$, or the value is zero for any duration less than $r_{n,m}$. This means each buyer's charging request cannot be partially fulfilled.
\begin{exmp}
As shown in Fig.~\ref{example}, the charger sharing market has two sellers and one buyer. Seller 1 and seller 2 submit their locations $l_1$ and $l_2$ and types (13:00, 17:00, \$1.5) and (15:00, 19:00, \$1) to the market, respectively. Based on these information, Buyer 1 has a type $\theta_{1,1} =$ (12:00, 16:00, 2, \$4) for seller 1 and a type $\theta_{1,2} =$ (16:00, 20:00, 3, \$5) for seller 2. The feasible start charging time of buyer 1 is constrained by the time windows of sellers and also
itself. In this case, buyer 1's  start charging time is either 13:00 or 14:00 on seller 1, and 16:00 on seller 2. The feasible charging time windows of buyer 1 at both seller 1 and seller 2 are highlighted in grey.
\end{exmp}
\begin{figure}[!htbp]
\centering
\includegraphics[width=8.75cm]{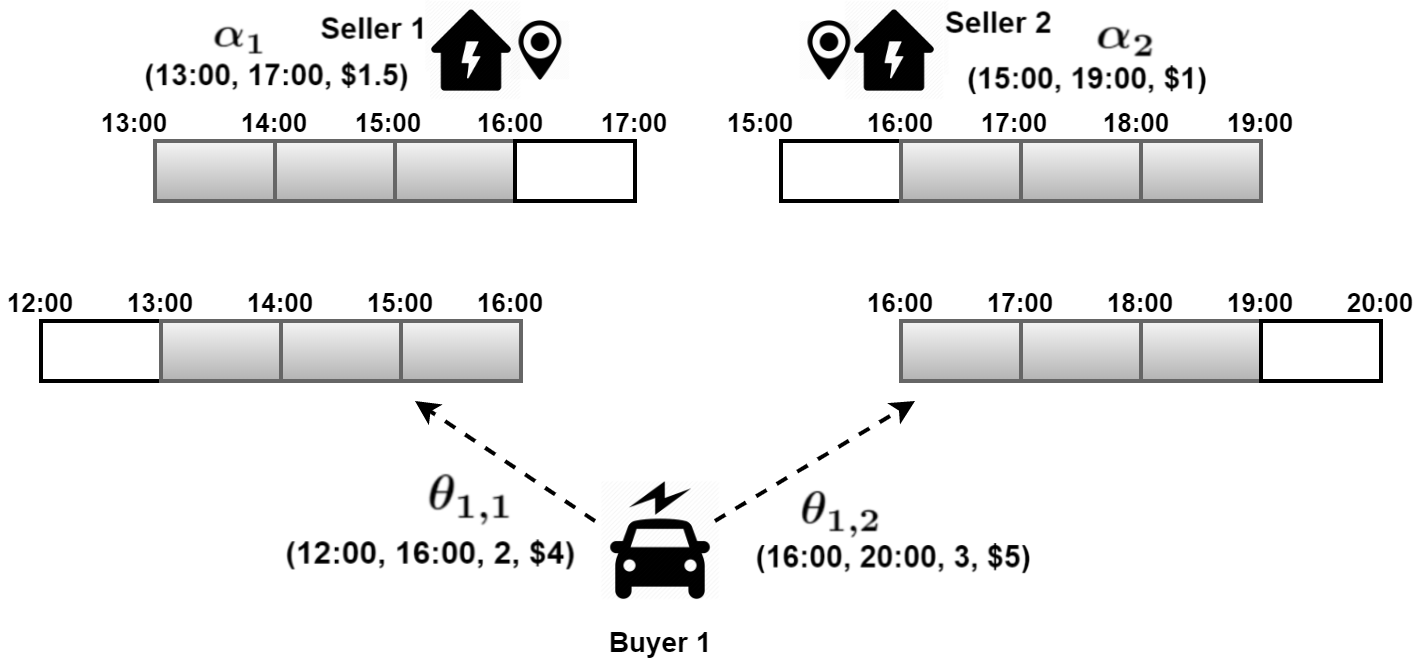}
\caption{Example of a CSS problem}
\label{example}
\end{figure} 
Given the types from sellers and buyers, the solution to the CSS problem is a \textit{schedule} which can be represented by a matrix ${L}\in \mathbb{R}^{|\mathcal{M}|\times|\mathcal{N}|}$, where each element $l_{n,m}$ denotes the start charging time of buyer $n\in \mathcal{N}$ at seller $m\in\mathcal{M}$.
%$\boldsymbol{L}= [l_{n,m}| n\in N, m\in M]$, which defines the allocated start charging time $l_{n,m}$ of buyer $n\in N$ at seller $m\in M$. 
Let $l_{n,m} \in \mathbb{R}^{+}$ if $n$ is allocated to $m$, and $l_{n,m}=-1$ otherwise. A \textit{feasible} schedule ${L}$ in our problem setting is defined as follows:
\begin{definition}
A schedule ${L}$ is feasible if for every element $l_{n,m}\geq 0$ it satisfies the following constraints:
\begin{enumerate}[label=(\roman*)]
\item Buyer $n$ cannot start before its arriving time, i.e.,
        $l_{n,m}\geq {a}_{n,m}$,
        \label{i}
    \item Buyer $n$ cannot finish after its departure time, i.e., $l_{n,m}\leq {d}_{n,m}-{r}_{n,m}$,
    \label{ii}
    \item Buyer $n$ can only start once, i.e., $\forall m,m'\in\mathcal{M}$: if $l_{n,m}\neq -1$ and $l_{n,m'}\neq -1$ then $m=m'$,
     \label{iii}
    \item If any buyers $n$ and $n'$ are allocated to the same seller, either $n$ must be finished before the start charging time of $n'$ or $n'$ must be finished before the start charging time of $n$, i.e., $\forall n,n'$: if $l_{n,m}\neq -1$ and $l_{n',m}\neq -1$, then $l_{n,m} + {r}_{n,m} \leq l_{n',m} + H\cdot(1-Y_{n,n',m})$ and $l_{n',m} + {r}_{n',m} \leq l_{n,m} + H\cdot Y_{n,n',m}$\footnote{$Y_{n,n',m}$ is the disjunctive variable: $Y_{n,n',m} = 1$ when $n$ is scheduled before $n'$ on $m$ and $Y_{n,n',m} = 0$ $n'$ is first. $H$ is a large positive constant which is used for the linearization of the logical constraint ``if'' \cite{rardin1998optimization}.},
     \label{iv}
    \item If buyer $n$ is allocated to seller $m$, the charging time should be within $m$'s available time window, i.e., ${s}_m\leq l_{n,m}\leq {e}_m - {r}_{n,m}$ ,
    \label{v}
    \item If buyer $n$ is allocated to seller $m$, the buyer's value cannot be less than the cost of the seller, i.e., $\forall n,m$: if $l_{n,m}\neq -1$, then ${{v}}_{n,m} \geq {r}_{n,m}\cdot{c}_m$.
    \label{vi}
\end{enumerate}
\label{defination}
\end{definition}
Let $\mathbbm{1}_{l_{n,m}\neq-1}$ be the indicator function that equals 1 if $l_{n,m}\neq-1$ is true and 0 if otherwise.
The social welfare of a feasible schedule ${L}$ is then defined as the difference between the sum of the allocated buyers' values and the sum of the allocated sellers' costs:
%\begin{equation}
 %\sum_{i=1}^n\sum_{j=1}^m  X_{i,j} \cdot (\hat{v}_{i,j} - r_{i,j} \cdot \hat{c}_j)
%\end{equation}
\begin{align}
\displaystyle {\sum_{n\in \mathcal{N}}\sum_{m\in \mathcal{M}}}
  ({v}_{n,m} - {r}_{n,m} \cdot {c}_n)\cdot \mathbbm{1}_{l_{n,m}\neq-1}
\end{align}
which is the value the scheduler aims to maximize.

Since we consider a game-theoretic setting, sellers and buyers are modeled as self-interested agents and may provide \textit{untrue} types if these are in their best interest. Here we need to define notations for reported types. Let
$\hat{\alpha}_m = ( \hat{s}_m, \hat{e}_m, \hat{c}_m )$ denotes the reported type of seller $m$ and $\hat {\boldsymbol{\theta}}_n = (\hat{\theta}_{n,m}: m\in \mathcal{M})$ 
denote buyer $n$'s reported types, in which a single reported type is represented as $\hat{\theta}_{n,m}= ( \hat{a}_{n,m}, \hat{d}_{n,m}, \hat{r}_{n,m}, \hat{v}_{n,m} )$. Let $\hat{\boldsymbol{\alpha}}_{-m}$ and $\hat{\boldsymbol{\theta}}_{-n}$ denote all seller and buyer reported types except that of $m$ and $n$, where ${\hat{\boldsymbol{\alpha}}} = (\hat{\alpha}_m, \hat{\alpha}_{-m})$ and ${\hat{\boldsymbol{\theta}}} = (\hat{\boldsymbol{\theta}}_{n}, \hat{\boldsymbol{\theta}}_{-n}$). Note that $\hat{\alpha}_m$ and $\hat{\boldsymbol\theta}_{n}$ may not be equal to $\alpha_m$ and ${\boldsymbol\theta}_{n}$. In this paper, we assume restricted reports which means
sellers cannot report an earlier service start time nor a later service end time and buyers cannot report an earlier arrival nor a later departure time. More formally, we assume $\hat{s}_m \geq {s}_m$ and $\hat{e}_m \leq e_m$ for all seller $m \in \mathcal{M}$ and for each buyer $n$, $\hat{a}_{n,m} \geq a_{n,m}$ and $ \hat{d}_{n,m} \leq d_{n,m}$, $\forall m\in \mathcal{M}$.

As the scheduler is unaware of the private information of sellers and buyers. Given the reported types, 
the scheduler might make sub-optimal decisions. 
Therefore, to efficiently allocate buyers to sellers, we propose a price-based iterative double auction to solve the CCS problem in charger sharing markets.
%the type of buyer $n$ is a vector $\boldsymbol \theta_n = (\theta_{n,m}: m\in M)$. Each element $\theta_{n,m}= \langle a_{n,m}, d_{n,m}, r_{n,m}, v_{n,m} \rangle$ is a tuple, denoting buyer $i$'s arriving time $a_{i,j}\in T$, departure time $d_{i,j}\in T$, required charging time units $r_{i,j}\in \mathbb{R}_{\geq 0}$, and value $v_{i,j}\in \mathbb{R}_{\geq 0}$ for a particular seller $j$. This is a buyer's private information. 

\section{Price-based Iterative Double Auction}
\label{IDA}
In this section, we propose a price-based iterative double auction (P-IDA) for the CSS problem.
Compared with one shot auctions, such as VCG~\cite{samadi2011optimal,de2015market,gerding2013two}, the iterative bidding structure of the proposed auction promises reduced computation at the auctioneer side and partial revelation of the private information at buyers' and sellers' sides~\cite{ parkes2001iterative,nisan2007algorithmic}. In addition, compared with single-sided multi-round auctions, such as those proposed in~\cite{parkes2006iterative,wang2011due, wang2013service}, the two-sided structure of P-IDA facilitates trade between two groups in one market,
which is more efficient than combining several single-sided auctions~\cite{xia2005solving}. 
%In this mechanism, a scheduler acts as an auctioneer to perform an iterative double auction according to the buying price from buyers and selling price from sellers. In this way, the auctioneer determines the allocation and trading price during each iteration. 
\subsection{The P-IDA structure}
\begin{figure}[!htbp]
\centering
\includegraphics[width=7.5cm]{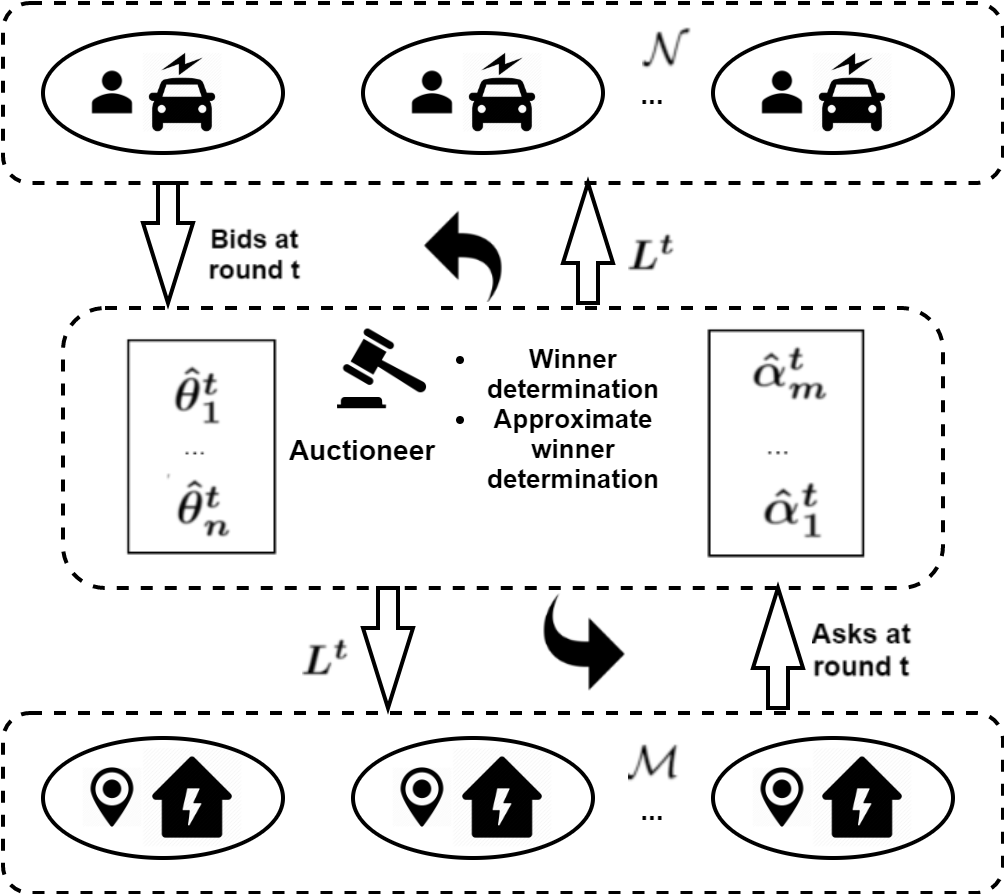}
\caption{The P-IDA structure}
\label{structure}
\end{figure} 
%A high level structure of the P-IDA mechanism is shown in Fig.~\ref{structure}. 
As shown in Fig.~\ref{structure},
P-IDA involves a group of sellers $\mathcal{M}$, a group of buyers $\mathcal{N}$ and an auctioneer. The auction proceeds in rounds and during each round $t$, sellers submit \textit{asks} to the auctioneer which include the charging time windows they offer for sale and their ask prices. Let $\hat{\alpha}^t_m = (\hat{s}_m, \hat{e}_m, p^{s,t}_m)$ denote the ask of seller $m$ at round $t$, where $p^{s,t}_m$ indicates the ask price in round $t$ for a time unit. At the same time, buyers submit \textit{bids} to the auctioneer which specify their available charging time windows, required charging duration, and their bid prices. Let $\hat{\theta}^t_{n,m} = (\hat{a}_{n,m}, \hat{d}_{n,m}, \hat{r}_{n,m},  \hat{r}_{n,m}\cdot p^{b,t}_{n,m})$ denotes the bid of buyer $n$ for seller $m$ at round $t$, where $p^{b,t}_{n,m}$ indicates the bid price of a time unit in round $t$. Thereafter, the auctioneer evaluates the bids and asks to determine a provisional feasible schedule $L^t$ with the objective of maximizing the social welfare. The auctioneer can also use an approximation algorithm for winner determination when the size of the model is too large to be optimally computed in a reasonable time. After receiving the provisional schedule from the auctioneer at round $t$, sellers and buyers will update their time unit prices $p^{s,t}_m$ and $p^{b,t}_{n,m}$ and submit the updated ones to the auctioneer for the next bidding round. The auction terminates when the termination conditions are satisfied.

In the following subsection, we formulate the winner determination model for the auctioneer to compute a social welfare maximizing provisional schedule during each round. 
\subsection{Winner determination}
The auctioneer solves the winner determination model during each round, computing a provisional schedule to maximize the social welfare.
The winner determination model in round $t$ is formulated as:
\begin{align}
    \max_{L^t}
    \displaystyle {\sum_{n\in \mathcal{N}}\sum_{m\in \mathcal{M}}} \hat{r}_{n,m}\cdot( p^{b,t}_{n,m} -  p^{s,t}_{m})\cdot\mathbbm{1}_{l^t_{n,m}\neq-1}
    \label{2}
\end{align}
with respect to the feasible schedule constraints (see Definition~\ref{defination} for details). $L^t$ is a provisional schedule at round $t$.
%$\Gamma^t$ is a set of all feasible schedules in round $t$. 
As multiple optimal solutions may exist, tie-breaking rules are needed for the auctioneer to decide which of the optimal solutions is selected. In this paper, ties are broken in favour of maximizing the number of trades and then at random.
\subsection{Approximate winner determination}
The winner determination model can be solved using standard integer programming optimization packages, such as ILOG CPLEX. 
However, due to the NP-completeness of the winner determination problem, computing the optimal schedule by using CPLEX is achievable in infeasible computational time for larger problem instances. Thus, the auction is not applicable to realistic scenarios in which agents require a timely schedule. Therefore, to further test the efficiency of the proposed double auction in large scale settings and to foster the use of P-IDA in real practical applications, an algorithm which is based on meta-heuristic simulated annealing (SA)~\cite{kirkpatrick1983optimization} is designed to solve the winner determination model during each bidding round for the auctioneer.
  
The SA algorithm starts by randomly generating an initial schedule based on the bids and asks submitted by buyers and sellers in each bidding round. Then the SA proceeds in several iterations. 
At the beginning of each SA iteration, multiple schedules are generated by a random permutation of the initial schedule. A schedule is always selected if it has a superior social welfare than that of the initial schedule, otherwise a schedule is accepted with a given decreasing possibility which is based on Boltzmann distribution~\cite{kirkpatrick1983optimization}. In the final iteration of SA, the schedule with the highest social welfare will be the provisional schedule for buyers and sellers in the current bidding round. Ties are broken at random. This ``P-IDA-SA" approach maintains the same incentive for buyers and sellers to follow myopic best-response bidding strategies~\cite{parkes2000iterative}.
%However, 

As the winner determination model takes the bid prices and the ask prices as inputs, we then specify the bidding and price updating rules for buyers and sellers in the following subsection.
\subsection{Bidding and price updating rules}
Initially, buyers and sellers receive from the auctioneer a minimum and a maximum time unit price for charging. The minimum price is a reserve price which reflects the basic charger installation cost and electricity fees. Any bid prices lower than such price are deemed as invalid and will be rejected by the auctioneer.  The maximum time unit price is a highest reference value capturing a proportion of property tax, charger construction cost, maintenance cost and electricity fees. Any ask prices higher than such reference will be rejected by the auctioneer. Let $b_{min}\in \mathbb{R}^{+}$ be the minimum time unit price for all bids and $a_{max}\in \mathbb{R}^{+}$ be the maximum time unit price for all asks, with $b_{min} < a_{max}$. Based on $b_{min}$ and $a_{max}$, buyers and sellers set up their first round bid prices and ask prices.

\subsubsection{Buyer's and sellers' bidding rules}
At the beginning of round $t-1 (t>1)$, buyers need to compute a set of utility-maximizing bids among all their bids based on the bid prices and values. To compute such a set, a buyer $n\in\mathcal{N}$ solves the utility maximization problem: 
\begin{align}
\max_{\hat{\theta}_{n,m}^{t-1}\in {\hat{\boldsymbol{\theta}}^{t-1}_{n}}}  (v_{n,m} - \hat{r}_{n,m}\cdot p^{b,t-1}_{n,m})\cdot{\mathbbm{1}_{l^{t}_{n,m}\neq-1}}
\label{3}
\end{align}
and obtains a set of bids which equally maximize his or her utility, where $p^{b,t-1}_{n,m}$ is the bid price of buyer $n$ for seller $m$ at round $t-1$. That is, for any two bids $\hat{\theta}^{t-1}_{n,1}$ and $ \hat{\theta}^{t-1}_{n,2}$ in the utility maximization set: $v_{n,1} - \hat{r}_{n,1}\cdot p^{b,t-1}_{n,1} = v_{n,2} - \hat{r}_{n,2}\cdot p^{b,t-1}_{n,2}$. After obtaining utility maximization bids at round $t-1$, the buyer can have two bidding strategies: he or she can randomly pick one (\textit{single-bid} bidding strategy) or join them together as an XOR-bid (\textit{xor-bid} bidding strategy) \cite{nisan2006bidding} and submits it to the auctioneer. An XOR-bid: 
\begin{align*}
  \hat{\theta}^{t-1}_{n,1}\oplus \hat{\theta}^{t-1}_{n,2}\oplus \dots   \oplus\hat{\theta}^{t-1}_{n,m}
\end{align*}indicates that:
\begin{itemize}
    \item these bids equally maximize buyer $n$'s utility at round $t-1$ based on the bid prices and his or her values,
    \item at most one $\hat{\theta}^{t-1}_{n,k}$ can be accepted by the auctioneer.
\end{itemize}

Meanwhile, at the beginning of round $t-1$, a seller $m\in \mathcal{M}$ submits the ask to the auctioneer if this seller has available charging time.

After receiving the bids from buyers and asks from sellers at round $t-1$, the auctioneer solves the winner determination problem and sends the provisional schedule $L^{t-1}$ back to the buyers and sellers. At the beginning of round $t$, buyers and sellers need to update their bid prices and ask prices based on the provisional schedule $L^{t-1}$ at round $t-1$. The price updating rules of buyers and sellers are given as follows.
\subsubsection{Buyer's and seller's price updating rules}
If a buyer is not included in the provisional schedule $L^{t-1}$, this buyer has the following two price updating options in round $t$:
\begin{itemize}
    \item This buyer can increase his or her bid prices by $\epsilon>0$ on sellers that this buyer bid for at round $t-1$ or rounds before $t-1$. Here $\epsilon $ is the minimum bid-increment in the auction for buyers. As buyers are assumed to be rational in maximizing their utilities, they do not bid with an increment more than $\epsilon$.
    \item This buyer can keep his or her bid prices unchanged or make an increment less than $\epsilon$. It happens when the utility of all other bids are non-positive. In this case, the auctioneer will consider that this buyer has entered into the final bid status and the buyer is forbidden from increasing the bid prices of those bids in the future rounds. 
\end{itemize}
If a buyer is included in the provisional schedule $L^{t-1}$, this buyer can keep its bidding price unchanged at round $t$. That is, the buyer is allowed to repeat the same bid at round $t$. Note here, for a buyer with the xor-bid bidding strategy, the repeated bid at round $t$ should be the awarded one in XOR-bid at round $t-1$. However, the auctioneer 
can also choose to allow the buyer to repeat its XOR-bid in future rounds until the auction terminates. The purpose for this \textit{xor-bid-repeating} is to boost the efficiency of the auction. 
After updating the bid prices, buyers recompute their utility-maximizing bids by using~(\ref{3}). If there are multiple utility-maximizing bids,
this buyer can randomly pick one or joint them as XOR-bid for the next bidding round $t$. 
\begin{exmp}
Consider a buyer $n\in \mathcal{N}$ with xor-bid bidding strategy submits a bid $\hat{\theta}^1_{n,k}$ to the auctioneer in the first round. Suppose after computing the winner determination problem by the auctioneer, he or she is not included in $L^1$. Buyer $n$ will increase $p^{b,1}_{n,k}$ to $p^{b,1}_{n,k} + \epsilon$. After updating the bid price, buyer $n$ recomputes the utility maximization bids and obtains two bids $\hat{\theta}^1_{n,k}$ and $ \hat{\theta}^1_{n,k'}$. Thereafter, in the second round, he or she submits  $\hat{\theta}^2_{n,k'} \oplus\hat{\theta}^2_{n,k}$ to the auctioneer, where $\hat{\theta}^2_{n,k'}= (\hat{a}_{n,k'}, \hat{d}_{n,k'}, \hat{r}_{n,k'}, p^{b,2}_{n,k'})$ and $\hat{\theta}^2_{n,k}= (\hat{a}_{n,k}, \hat{d}_{n,k}, \hat{r}_{n,k}, p^{b,2}_{n,k})$. In this case, $p^{b,2}_{n,k'} = p^{b,1}_{n,k'}$ and $p^{b,2}_{n,k}= p^{b,1}_{n,k}+w\cdot\epsilon$.  If $0\leq w< 1$, buyer $n$ cannot increase his or her bid price on $k$ in the future round. As buyers are rational utility maximization agents, they will not bid with a $w>1$.
\end{exmp}
Given a provisional schedule $L^{t-1}$, 
if a seller sells all of his or her available charging time at round $t-1$, the seller repeats his or her ask at round $t$. 

If, on the other hand, a seller $m\in\mathcal{M}$ still has available charging time at round $t-1$, this seller has the following two ask price updating options in round $t$:
\begin{itemize}
    \item This seller can decrease the time unit price by $\epsilon >0$ at round $t$, where $\epsilon$ is the minimum ask-decrement in the auction for sellers. As sellers are assumed to be rational in maximizing their utilities, they do not bid with an decrement more than $\epsilon$.
    \item This seller can keep the ask price unchanged or make a decrement less than $\epsilon$. This happens when $p^{s,t-1}_{m}=c_m $ or $p^{s,t-1}_{m}-c_m \leq \epsilon$. In this case, the seller is not allowed to decrease his or her ask prices in the future rounds. 
\end{itemize}

\subsection{Termination condition and trading prices}
Once the auctioneer receives the bids and asks from buyers and sellers, the auctioneer checks the termination condition. 
The auction terminates when:
\begin{itemize}
    \item (T1) all
buyers and sellers submit the same bids and asks in two consecutive rounds.
\end{itemize}
On termination, the provisional schedule becomes the final schedule, and buyers pay their final bid prices to the auctioneer. Each seller's reimbursement is the sum of the allocated buyers' bid prices. 

Let $t^T$ denotes the final bidding round. Given the schedule and the prices at $t^T$, the utility of a buyer $n$ is defined as:
\begin{align}
    u^{b}_{n}(L^{t^T}, \hat{\boldsymbol{\theta}}^{t^T}_n) = \begin{dcases} 
                v_{n,m}-p^{b,t^T}_{n,m}\cdot\hat{r}_{n,m} & \text{if }  {\mathbbm{1}_{l^{t^T}_{n,m}\neq-1} =1}, \\
               0 & \text{otherwise}. \\
               \end{dcases} 
               \label{5}
\end{align}
Accordingly, a seller $m$'s utility is defined as:
\begin{align}
    u^{s}_{m}(L^{t^T},\hat{\alpha}^{t^T}) =  \begin{cases}
                \sum\limits_{n\in N}  \hat{r}_{n,m}\cdot(p^{b,t^T}_{n,m}-c_{m}) & \text{if }  {\mathbbm{1}_{l^{t^T}_{n,m}\neq-1} =1}, \\
               0 & \text{otherwise}.  \\
               \end{cases} 
               \label{6}
\end{align}

\subsection{Algorithm implementation}
\begin{algorithm}
\caption{Price-based iterative double auction}\label{euclid}
\label{al}
\hspace*{\algorithmicindent} \textbf{Input: $\boldsymbol{l}_m , \hat{\alpha}_m, \hat{\boldsymbol{\theta}}_{n}, w (0<w\leq1)$}, $\forall m\in \mathcal{M}, \forall n\in \mathcal{N}$  \\
    \hspace*{\algorithmicindent} \textbf{Output: ${L}^t$ }
\begin{algorithmic}[1]
%\Output{1}{}
\State $t \gets 1$;
\State Initialize $p^{b,1}_{n,m}$, $p^{s,1}_{m},\epsilon$;
\State flag $\gets 0$;
\While{flag $=0$}
\ForAll{$n\in \mathcal{N}$}
\If{$t>1\wedge$ $n$ is not included in $L^{t-1}$}
\ForAll{$\hat{\theta}^{t-1}_{n,m}$ submitted at round $t-1$}
\State $p^{t}_{n,m}=p^{t-1}_{n,m}+w\cdot\epsilon$;
\EndFor
\EndIf
\State\textbf{end}
\State Solve utility-maximizing bids by~(\ref{3});
\State Send bids to the auctioneer;
\EndFor
\State\textbf{end}
\ForAll{$m\in \mathcal{M}$}
\If{$t>1\wedge$ $m$ has available time in round $t-1$}
\State $p^{t}_m = p^{t-1}_m -w\cdot\epsilon$;
\EndIf
\State\textbf{end}
\State Send the ask to the auctioneer;
\EndFor
\State\textbf{end}
\If {T1 is satisfied}
\State $flag \gets1$;
\State break;
\EndIf
\State\textbf{end}
\State The auctioneer computes ${L}^t$ by (\ref{2}) with respect to feasible schedule constraints;
\State The auctioneer sends the provisional schedule ${L}^t$ to buyers and sellers;
\State $t\gets t+1$;
\EndWhile
\State\textbf{end}
\State The auctioneer collects payment $p^{b,t}_{n,m}$ from each buyer and reimburse to the corresponding sellers;
%\EndOutput
\end{algorithmic}
\end{algorithm}
The entire P-IDA is described in details in Algorithm~\ref{al}. The auctioneer initializes the minimum bid price for buyers and maximum ask price for sellers (i.e., $b_{min}$ and $a_{max}$) and also the increment or decrement $\epsilon$ (Line 2). For such an initialization, we can choose any set of values that satisfy $b_{min}<a_{max}$. Before the bid starts, all sellers set up their initial ask prices $p^{s,1}_m$ with $p^{s,1}_m\leq a_{max}$. The auctioneer then broadcasts the locations and asks of the participating sellers to the buyers. After receiving this information, buyers compute their bids for each feasible seller. The initial bid price $p^{b,1}_{n,m}$ for each buyer should satisfy $p^{b,1}_{n,m}\geq b_{min}$.
Then
the auctioneer will iteratively compute the provisional schedule given the bids and asks until the termination condition is satisfied. In the first bidding round, each seller submits his or her asks to the auctioneer and each buyer obtains the utility-maximizing bids among his or her bids by solving~(\ref{3}) and sends the bid (single-bid or xor-bid) to the auctioneer. The auctioneer generates a provisional feasible schedule ${L}^t$ by solving the winner determination problem and sends $L^t$ back to the buyers and sellers. 
%For large-scale settings, the winner determination model during each round is solved by a meta-heuristic algorithm which is based on simulated annealing.
After receiving the provisional schedule ${L}^t$, the unscheduled buyers and the sellers who still have available charging time will update their prices. After receiving the updated prices, the auctioneer checks the termination condition. If the termination condition T1 is not satisfied, the auctioneer takes the updated prices from buyers and sellers as inputs and computes the provisional schedule for the next bidding round. When the auction terminates, the auctioneer determines the final payments
and reimbursements for buyers and sellers. The last round provisional schedule will be the final schedule.

Compared with one-shot VCG auction, the proposed P-IDA has two main advantages. Firstly, in solving the CSS problem, P-IDA has improved computational properties. The VCG is a sealed bid auction which motivates agents to submit their complete valuations truthfully and computes optimal solutions~\cite{ausubel2006lovely}. However, despite its theoretical elegance, it has its limitations in terms of implementation~\cite{sandholm2002algorithm}. Specifically, from the auctioneer's side in solving the CSS problem, the implementation of VCG auction requires $|\mathcal{M}|+|\mathcal{N}|+1$ NP-complete optimization problems. Therefore, the computational cost of the VCG auction is prohibitively expensive if the auction is applied to a CSS problem with nontrivial-size.  However, the iterative bidding structure of P-IDA can help to distribute the computation in the auction~\cite{parkes2001auction, parkes2000iterative}. Although the winner determination problem during each bidding round remains NP-complete, the problem instances in P-IDA are much smaller than that in VCG as buyers only bid for a small subset of sellers in each round, which largely reduces the computational complexities~\cite{parkes2001iterative}.
In addition, P-IDA preserves the privacy of buyers and sellers. In the VCG, agents need to submit their complete valuations~\cite{parkes2001auction} to compute a final schedule. However, in P-IDA, agents are not required to submit complete and exact information about their private information.
\section{Theoretical analysis}
\label{tr}
In this section, we analyze the economic properties of P-IDA. We prove that P-IDA is budget-balanced and individually rational. Budget balance ensures that the auction requires no outside subsidy which incentives the auctioneer (or the platform) to organize such an auction.  Individual rationality implies that participants are never worse off by participating in the auction which attracts both buyers and sellers to participate in this auction. In addition, we prove that under this auction the weakly dominant strategies for buyers and sellers are to truthfully reveal their charging time constraints.
This property simplifies the strategic space of the participating buyers and sellers, and enables more efficient allocations.
\begin{rem}
P-IDA is budget-balanced because the total payments collected from the buyers is equal to the total reward assigned to the sellers.
\end{rem}
\begin{theo}
P-IDA is individually rational for all participating buyers and sellers.
\end{theo}
\begin{proof}
This theorem will be established separately for the buyers and the sellers. Note that not participating in the auction leads to a zero utility for any buyers and sellers as they will not receive or offer any charging services nor they will pay any kind of fee. 
For an arbitrary buyer $n\in \mathcal{N}$, upon termination of the auction, if he or she is not included in the final schedule, then the utility of buyer $n$ is zero: $u^b_n =0$ (see Equation (\ref{5})).

If on the other hand, buyer $n\in \mathcal{N}$ is allocated to an arbitrary seller $m\in \mathcal{M}$ when the auction terminates, then the utility of buyer $n$ is $v_{n,m}-p^{b,t^T}_{n,m}\cdot \hat{r}_{n,m}$, where $p^{b,t^T}_{n,m}$ is buyer $n$'s last round bid price. Since $p^{b,t^T}_{n,m}\cdot \hat{r}_{n,m} \leq v_{n,m}$ always holds, it follows that: 
\begin{align*}
   u^b_n = v_{n,m}-p^{b,t^T}_{n,m}\cdot \hat{r}_{n,m} \geq 0.
\end{align*}
Therefore, it is concluded that P-IDA is individually rational for all participating buyers.   

If an arbitrary seller $m$ is not included in the final schedule, then $u^s_m = 0$ (see Equation (\ref{6})). If, on the other hand, seller $m$ is included in the final schedule, to satisfy the feasible schedule constraint~\ref{vi}, we have ${p}^{b,t^T}_{n,m}- p^{s,t^T}_m\geq 0$ always holds, where $p^{s,t^T}_m$ is seller $m$'s ask price in the last round. Since $p^{s,t^T}_m \geq c_m$, it follows that
\begin{align*}
   {p}^{b,t^T}_{n,m}- c_m \geq {p}^{b,t^T}_{n,m}-p^{s,t^T}_m \geq 0.
\end{align*}
Therefore, 
\begin{align*}
   u^s_m =  \sum_{n\in \mathcal{N}}  (p^{b,t^T}_{n,m}-c_{m})\cdot \hat{r}_{n,m} \geq {p}^{b,t^T}_{n,m}-c_m \geq 0.
\end{align*}
 Thus we conclude that P-IDA is individually rational for all participating sellers.
\end{proof}
To prove truthfulness of P-IDA, we need to show that for an arbitrary buyer $n\in \mathcal{N}$ and seller $m\in \mathcal{M}$, regardless of other agents' bids and asks, $n$ and $m$ can never gain by submitting untruthful charging time constraints, subject to the restricted reports assumption. 

We first prove that this mechanism is truthful for buyers in reporting their charging time constraints. We start by showing that if truthful reveals of $a_{n,k}, d_{n,k}, r_{n,k}$ do not cause buyer $n$ to be allocated to seller $k$, then strategically reveals of these variables can never cause the buyer to be allocated to seller $k$.
We break this part of proof into two lemmas, first showing that it holds for the arriving time regardless of other variables (Lemma 1), and then for departure time and required charging duration (Lemma 2). According to the single-minded assumption, a buyer will not report a length of charging shorter than the true length which means for each buyer $n$, $\hat{r}_{n,m} \geq r_{n,m}$, $\forall m\in \mathcal{M}$.
\begin{lemma}
In P-IDA, it is a weakly dominant strategy for buyers to truthfully report their arriving times.
\end{lemma}
\begin{proof}
Let $\boldsymbol{\hat{\theta}}'_n = (\hat{\theta}'_{n,k}: k\in \mathcal{M})$, where each element $\hat{\theta}'_{n,k} = ({a}_{n,k}, \hat{d}_{n,k}, \hat{r}_{n,k},\hat{r}_{n,k}\cdot{p}^{b,t}_{n,k})$ indicates that buyer $n$ truthfully reports his or her arriving time $a_{n,k}$.
%We consider an arbitrary reported type $\hat{\theta}_{n,k} = (\hat{a}_{n,k}, \hat{d}_{n,k}, \hat{r}_{n,k},\hat{v}_{n,k})$ of buyer $n$ for seller $k$. 
To prove this lemma we need to show that 
\begin{align}
    u^b_n (\boldsymbol{\hat{\theta}}'_n, \boldsymbol{\hat{\theta}}_{-n},\boldsymbol{\theta}_n) \geq u^b_n (\boldsymbol{\hat{\theta}}_n, \boldsymbol{\hat{\theta}}_{-n},\boldsymbol{\theta}_n)
    \label{10}
\end{align}
holds for all buyers.
First we randomly select one bid $\hat{\theta}'_{n,k}$ from buyer $n$ and assume that this is the bid upon termination of the auction. Then we consider the following two cases:

Case 1:
If $l_{n,k}\neq -1$, which means buyer $n$ is allocated to seller $k$ in the final schedule, then $l_{n,k} \geq a_{n,k}$. Given restricted reports assumption, we have $\hat{a}_{n,k} \geq a_{n,k}$ and the following two possibilities:

1). $a_{n,k}\leq l_{n,k}<\hat{a}_{n,k}$, then buyer $n$ cannot be allocated by reporting $\hat{\theta}_{n,k} = (\hat{a}_{n,k}, \hat{d}_{n,k}, \hat{r}_{n,k},\hat{v}_{n,k})$. Hence, we can derive:
\begin{align*}
    u^b_n ({\hat{\theta}'}_{n,k}, \boldsymbol{\hat{\theta}}_{-n},\boldsymbol{\theta}_n) \geq u^b_n ({\hat{\theta}}_{n,k}, \boldsymbol{\hat{\theta}}_{-n},\boldsymbol{\theta}_n)= 0.
\end{align*}

2). $a_{n,k}\leq \hat{a}_{n,k}\leq l_{n,k}$, then buyer $n$ can be allocated by reporting $\hat{\theta}_{n,k} = (\hat{a}_{n,k}, \hat{d}_{n,k}, \hat{r}_{n,k},\hat{v}_{n,k})$. Thus:
\begin{align*}
    u^b_n ({\hat{\theta}'}_{n,k}, \boldsymbol{\hat{\theta}}_{-n},\boldsymbol{\theta}_n) = u^b_n ({\hat{\theta}}_{n,k}, \boldsymbol{\hat{\theta}}_{-n},\boldsymbol{\theta}_n)\geq 0.
\end{align*}

Case 2:
If $l_{n,k}=-1$, which means buyer $n$ is not included in the final schedule. We prove this case by using contradiction. First, we assume  buyer $n$ is included in the final schedule by reporting $\hat{\theta}_{n,k}$. Let $\hat{l}_{n,k}$ be the allocated starting time. Therefore, we have $\hat{a}_{n,k} \leq \hat{l}_{n,k} \neq -1$. Based on the restricted reports assumption, the following inequality holds:
\begin{align*}
    a_{n,k}\leq \hat{a}_{n,k}\leq \hat{l}_{n,k}.
\end{align*}
It means that buyer $n$ will be allocated to $k$ in the final schedule by reporting $\hat{\theta}'_{n,k}$, a contradiction. Therefore, \begin{align*}
    u^b_n ({\hat{\theta}'}_{n,k}, \boldsymbol{\hat{\theta}}_{-n},\boldsymbol{\theta}_n) = u^b_n ({\hat{\theta}}_{n,k}, \boldsymbol{\hat{\theta}}_{-n},\boldsymbol{\theta}_n)= 0.
\end{align*}

As buyer $n$ and seller $k$ are an arbitrary selection, we can conclude that Equation (\ref{10}) holds for all buyers. 
\end{proof}
\begin{lemma}
In P-IDA, it is a weakly dominant strategy for buyers to truthfully report their departure times and required charging duration.
\end{lemma}
\begin{proof}
Let $\boldsymbol{\hat{\theta}}''_n = (\hat{\theta}''_{n,k}: k\in M)$, where each element $\hat{\theta}''_{n,k} = (\hat{a}_{n,k}, {d}_{n,k}, {r}_{n,k},{p}^{b,t}_{n,k})$ indicates that buyer $n$ truthfully reports his or her departure time $d_{n,k}$ and required charging time duration $r_{n,k}$. To prove this lemma, we need to show the following inequality holds for all buyers:
\begin{align}
     u^b_n (\boldsymbol{\hat{\theta}}''_n, \boldsymbol{\hat{\theta}}_{-n},\boldsymbol{\theta}_n) \geq u^b_n (\boldsymbol{\hat{\theta}}_n, \boldsymbol{\hat{\theta}}_{-n},\boldsymbol{\theta}_n). 
     \label{8}
     \end{align}
We randomly select one bid $\hat{\theta}''_{n,k}$ from buyer $n$ and assume that this is the submitted bid upon termination of the auction. Then we have the following two cases:

Case 1: $l_{n,k} \neq -1$, which means buyer $n$ is included in the final schedule by reporting $\hat{\theta}''_{n,k}$, thus we have:
\begin{align*}
    l_{n,k} \leq {d}_{n,k} -{r}_{n,k} \wedge l_{n,k} \leq e_k-{r}_{n,k}.
\end{align*}
According to the restricted reports and single-minded assumption, we have $\hat{d}_{n,k}-\hat{r}_{n,k} \leq {d}_{n,k} -{r}_{n,k}$. Therefore, the only possibility that buyer $n$ is included in the final schedule by reporting  $\hat{\theta}_{n,k}$ is that the following two inequalities are satisfied at the same time:
\begin{align*}
    l_{n,k} \leq \hat{d}_{n,k}-\hat{r}_{n,k}\leq {d}_{n,k} -{r}_{n,k},
\end{align*} and 
\begin{align*}
  l_{n,k}\leq e_k-\hat{r}_{n,k} \leq e_k-{r}_{n,k}.
\end{align*}
If (13) and (14) are satisfied, 
then:
\begin{align*}
    u^b_n ({\hat{\theta}''}_{n,k}, \boldsymbol{\hat{\theta}}_{-n},\boldsymbol{\theta}_n) = u^b_n ({\hat{\theta}}_{n,k}, \boldsymbol{\hat{\theta}}_{-n},\boldsymbol{\theta}_n)\geq 0.
\end{align*}
Otherwise, 
\begin{align*}
    u^b_n ({\hat{\theta}''}_{n,k}, \boldsymbol{\hat{\theta}}_{-n},\boldsymbol{\theta}_n) \geq u^b_n ({\hat{\theta}}_{n,k}, \boldsymbol{\hat{\theta}}_{-n},\boldsymbol{\theta}_n)= 0.
\end{align*}

Case 2: $l_{n,k}=-1$, which means buyer $n$ is not included in the final schedule by reporting
$\hat{\theta}''_{n,k}$. We prove this case by contradiction. We assume buyer $n$ is included in the final schedule by reporting $\hat{\theta}_{n,k}$. Then we have: 
\begin{align*}
    l_{n,k} \leq \hat{d}_{n,k} -\hat{r}_{n,k} \wedge l_{n,k} \leq e_k-\hat{r}_{n,k}.
\end{align*}
Given that $\hat{d}_{n,k} < d_{n,k}$ and $\hat{r}_{n,k} > r_{n,k}$, we derive:
\begin{align*}
   \hat{d}_{n,k} -\hat{r}_{n,k} \leq d_{n,k} -r_{n,k} \wedge e_k-\hat{r}_{n,k} \leq e_k-{r}_{n,k}.
\end{align*}
This immediately follows that:
\begin{align*}
  l_{n,k} \leq {d}_{n,k} -{r}_{n,k} \wedge l_{n,k} \leq e_k-{r}_{n,k}.  
\end{align*}
It means that buyer $n$ will be allocated to seller $k$ by reporting $\hat{\theta}''_{n,m}$, a contradiction.
Therefore, 
\begin{align*}
    u^b_n ({\hat{\theta}''}_{n,k}, \boldsymbol{\hat{\theta}}_{-n},\boldsymbol{\theta}_n) = u^b_n ({\hat{\theta}}_{n,k}, \boldsymbol{\hat{\theta}}_{-n},\boldsymbol{\theta}_n)= 0.
\end{align*}
As buyer $n$ and seller $k$ are an arbitrary selection, deriving from the above two cases, we can conclude that Equation~(\ref{8}) holds for all buyers.
%\begin{align*}
    % u^b_n (\boldsymbol{\hat{\theta}}''_n, \boldsymbol{\hat{\theta}}_{-n},\boldsymbol{\theta}_n) \geq u^b_n (\boldsymbol{\hat{\theta}}_n, \boldsymbol{\hat{\theta}}_{-n},\boldsymbol{\theta}_n).
%\end{align*}
\end{proof}
\begin{theo}
In P-IDA, it is a weakly dominant strategy for each buyer to truthfully report their available charging time windows and required charging time duration.
\end{theo}
\begin{proof}
The proof of this theorem follows directly from the above two lemmas. In Lemma 1 we proved that for each buyer, truthful reveal of his or her arriving time is a weakly dominant strategy and in Lemma 2 we proved that it is a weakly dominant strategy for each buyer to truthfully reveal his or her departure time and required charging duration. Therefore we conclude that the proposed P-IDA is truthful for all buyers in reporting their charging time constraints.
%Therefore, 
%\begin{align*}
  %   u^b_n (\boldsymbol{{\theta}}_n, \boldsymbol{\hat{\theta}}_{-n},\boldsymbol{\theta}_n) \geq u^b_n (\boldsymbol{\hat{\theta}}_n, \boldsymbol{\hat{\theta}}_{-n},\boldsymbol{\theta}_n)
%\end{align*}
%for all price-taking buyers.

\end{proof}
%Then we prove that P-IDA is truthfulness for sellers.
\begin{theo}
In P-IDA, it is a weakly dominant strategy for each seller to truthfully report their available time windows.
\end{theo}
\begin{proof}
To prove this theorem, we need to show that the following condition holds for all sellers:
\begin{align}
    u^s_m (\alpha_m, \hat{\alpha}_{-m}, \alpha_m) \geq u^s_m (\hat{\alpha}_m, \hat{\alpha}_{-m}, \alpha_m).
    \label{9}
\end{align}
Assume by contradiction this condition does not hold, which means there exists a buyer allocated to seller $m$ when $\hat{s}_m$ and $\hat{e}_m$ are untruthfully reported, but this buyer is not allocated to $m$ for a truthful report. Randomly select a buyer $n\in \mathcal{N}$ and assume $n$ is allocated to seller $m$ by reporting $\hat{\alpha}$. Based on the feasible schedule constraint~\ref{v}, we get:
\begin{align*}
    \hat{s}_k \leq l_{b,k} \leq \hat{e}-r_{b,k}.
\end{align*}
Given $\hat{s}_m\geq s_m$ and $\hat{e}_m \leq e_m$, we have the following inequality that always holds:
\begin{align*}
    {s}_k \leq l_{b,k} \leq {e}-r_{b,k}.
\end{align*}
This means that seller $m$ will have an allocated buyer $n$ by reporting $s_k$ and $e_k$, a contradiction.

As seller $m$ and buyer $n$ are an arbitrary selection, we can conclude that Equation~(\ref{9}) holds for all sellers.
\end{proof}
\section{Computational study}
\label{na}
In this section, we conduct a computational study to verify the performance of P-IDA in terms of efficiency, profit ratio and running time under different problem scales and various configurations. We first define the evaluation metrics, after that we describe the experiment settings and then analyze the experiment results.
\subsection{Experimental Evaluation Metrics}
The evaluation metrics are defined as follows:
\begin{itemize}
    \item \textbf{Efficiency} of scheduling, $eff(L)$, is measured as the ratio of the social welfare of the final schedule $L$ to the social welfare of the optimal schedule $L^*$
    \begin{align}
        eff (L) = \dfrac{\sum_{(i,j)\in L}(v_{i,j}-r_{i,j}\cdot c_{j})}{\sum_{(i,j)\in L^*}(v_{i,j}-r_{i,j}\cdot c_{j})}
    \end{align}
    where $L$ is the final schedule generated by the auction and $L^*$ is the optimal schedule which maximizes the social welfare.
\end{itemize}
\begin{itemize}
    \item \textbf{Profit ratio} of the auction, $pro (L)$, is measured as the sum of sellers' payoff in the final schedule $L$, as a fraction of the sum of the payoff in the optimal solution $L^*$ that maximizes the social welfare
    \begin{align}
        pro (L) =  \dfrac{\sum_{(i,j)\in L}(p^b_{i,j} - c_{j})\cdot r_{i,j}}{\sum_{(i,j)\in L^*}(v_{i,j}- r_{i,j}\cdot c_{j})}
    \end{align}
    where $p^b_{i,j}$ is the bid price of buyer $i$ for seller $j$ in the final schedule $L$ and $(b_{i,j} - c_{j})\cdot r_{i,j}$ is the payoff of seller $j$ in $L$.
    The profit ratio metric is designed to measure the degree to which the sellers make money by applying the auction.
\end{itemize}

\begin{itemize}
    \item \textbf{Running time} of the auction refers to the computation time needed to terminate the auction on a CSS problem instance.
\end{itemize}
%For the sake of clarity in analysis, we first study a small charger sharing market with 5 sellers and 5-20 buyers. Then we consider a large market with 20 sellers and 20-150 buyers. The testing groups are summarized in TABLE2.
\subsection{Settings}
%We consider a day-head CSS problem. 
The scheduling horizon is a 15 hour window on the next day from 07:00-22:00 and it is 
divided into 30 time units of half an hour.
For each seller $m$, the service start time $s_m$ is randomly drawn from a uniform distribution in the range of 07:00-14:00. We assume sellers offer at least 16 time units of charging service. Therefore, the available charging time unit is a uniform random number from $[16, \min\{30, 2\cdot(22-s_m)\}]$. Based on $s_m$ and the available charging time units, the service end time $e_m$ of seller $m$ can be directly generated. The time unit cost of each seller various from \$1 to \$2.5 with a step \$0.1.

For each buyer $n$, 
an arriving time $a_{n,m}$ is randomly drawn from a uniform distribution in the range of 07:00-21:30. According to a survey~\cite{santos2011summary}, the peak intervals of EV charging include 08:00-10:00, 12:00-14:00 and 18:00-20:00. 
Thus, we vary the relative proportions between
the buyers to represent varying levels of heterogeneity in the buyer
population. Specifically, our data is designed to ensure that 20\% of buyers are arrived during each of these three peak time intervals. 
A departure time $d_{n,m}$ of each buyer $n$ is randomly drawn from a uniform distribution in the range of $[a_{n,m}+1, \min\{a_{n,m} + 8, e_m\}]$ and required charging duration follows a uniform distribution over the interval $[1, \min\{d_{n,m}-a_{n,m}, \frac{C}{R}]$, where $C$ is a constant number, representing the maximum battery capacity (kWh) and $R$ is the constant charging
rate (kW) delivered at the charging station. For simplicity, we assume $C=80$ kWh for all buyers and $R = 10$ kW for all sellers. The time unit value (willingness to pay) of each buyer various from \$0.1 to \$5, with a step \$0.1. 
%According to a survey, the peak intervals of EV charging include 08:00-10:00, 12:00-14:00 and 18:00-20:00. We assume 20\% of buyers in each of these peak hours.
We assume the number of bids for each buyer is a random number from interval $[1, 0.4\cdot M]$, where $M$ is the number of sellers. 
 %The number of buyers are different during different time intervals. Based on, the peak intervals include 08:00-11:00, 4:00-7:00. In the designed data, 80\% of buyers would like to charge during that two time intervals. 40\% for each. Then, buyers are randomly distributed during other time intervals.
 For testing P-IDA, we have defined 15 groups of testing data which are shown in TABLE~\ref{table}. For each group, 10 instances are randomly generated.

\begin{table}[h]
\caption{ Testing data} 
\centering  
\resizebox{0.45\textwidth}{!}{% 
\begin{tabular}{p{0.08\textwidth}p{1.5cm}p{1.5cm}p{1.5cm}}
\hline
\hline
Group &Sellers &Buyers & \# Instances\\
\hline
1 & 4 & 5 & 10\\
2 & 4 & 10 & 10\\
3 & 4 & 15 & 10\\
4 & 4 & 20 & 10\\
5 & 5 & 5 & 10\\
6 & 5 & 10 & 10\\
7 & 5 & 15 & 10\\
8 & 5 & 20 & 10\\
9 & 6 & 5 & 10\\
10 & 6 & 10 & 10\\
11 & 6 & 15 & 10\\
12 & 6 & 20 & 10\\
13 & 20 & 50 & 10\\
14 & 20 & 100 & 10\\
15 & 20 & 150 & 10\\
\hline
\hline
\end{tabular} } 
\label{table}
\end{table}
\subsection{The performance of P-IDA}
The results of P-IDA over Groups 1-12 are compared with the \textit{First-Come-First-Served} (FCFS) algorithm in terms of efficiency. FCFS is a classic centralized scheduling algorithm which gives the priority to buyers who has an earlier arriving time. The optimal solutions for Groups 1-12 are computed by using Cplex\footnote{https://www.ibm.com/analytics/cplex-optimizer} solver. The FCFS algorithm and P-IDA are implemented in Python.
%\begin{figure}[h!]
%\centering
%\includegraphics[width=0.85\linewidth]{line0.eps}
%\caption{Efficiency comparison between FCFS and P-IDA with single-bid on Groups 1-12.}
%\label{fig3}
%\end{figure}

The optimal solutions and the efficiency results of P-IDA and FCFS over 12 groups with 
%We investigate the efficiency performance of the proposed P-IDA with single-bid over 12 groups of testing problem with 
$\epsilon = 0.2, a_{max} = 7$ and $b_{min} = 0.1$ are shown in Fig.~\ref{fig3}.
It is observed that P-IDA with single-bid can achieve on average 94\% of the  efficiency against the optimal solution among the 12 groups, outperforming the solution obtained by FCFS allocation policy (around 88\% on average out of 100\%). This makes sense because the FCFS policy allocates buyers to sellers according to the buyer arrival order rather than the social welfare. %In addition, we observe that the efficiency performance of P-IDA is stable when dealing with different sizes of the charger sharing market.

In Fig.~\ref{fig1} (a), we compared the efficiency of P-IDA on the three bidding rules, single-bid, xor-bid and xor-bid-repeating with different values of $\epsilon$. The results are averaged over the 12 groups.
It is demonstrated that, bidding with xor-bid-repeating has higher efficiency (on average 98\%) than single-bid (on average 94\%) and xor-bid (on average 97\%). However the cost is the increased computation time due to the complexity of solving the winner determination model. With xor-bid-repeating, the size of the winner determination model is usually larger than that of xor-bid and single-bid, and also the auctioneer may have to respect the hard XOR-bid constraints during each bidding round when solving the winner determination model. Therefore, the computation time of P-IDA is significantly increased with xor-bid-repeating, especially when the value of $\epsilon$ is small. As shown in Fig.~\ref{fig1} (b), with $\epsilon = 0.1$, the computation time of xor-bid-repeating is 6 times longer than xor-bid and single-bid. Also, the computation times of all these three bidding strategies are decreased by increasing the value of $\epsilon$. This is because $\epsilon$ controls the rate at which the prices of buyers and sellers are increased and decreased across rounds. Therefore, with lower value of $\epsilon$, the auction needs more rounds to clear the market which in turn, increases the computation time. As shown in Fig.~\ref{fig1} (c), with $\epsilon = 0.1$, the P-IDA with single bid, xor-bid and xor-bid-repeating will terminate after 59 rounds and down to around 6 rounds with $\epsilon =0.5$. %This indicates that, the running time is determined by not only the negotiation rounds but aslo the complexity of winner determination model.
In addition to that, more rounds requires buyers and sellers to submit more bids and asks which reveals more value and cost information. In theory, more private information revelation contributes to higher auction efficiency. Therefore, as shown in Fig.~\ref{fig1} (a), the efficiency of the auction with small value of $\epsilon$ usually achieves higher efficiency than a larger one. However, in our auction, even with a large value of $\epsilon$ ($\epsilon = 0.5$), bidding with single-bid can reach above 94\% efficiency and more than 97\% efficiency with xor-bid and xor-bid-repeating, averaged over the 12 groups. This verifies that the proposed double auction has high efficiency. 
%Fig.~\ref{f3} also indicates that xor-bid and single-bid strategy clear the market in real time, requiring less than one seconds.

%\begin{figure}[h!]
%\centering
%\includegraphics[width=0.85\linewidth]{line4.eps}
%\includegraphics[height=5cm,width=8.75cm]{jie4x12.eps}
%\caption{Profit ratio performance of the P-IDA averaged over 12 groups with different value of $a_{max}$.}
%\label{fii3}
%\end{figure}
Fig.~\ref{fii3} shows the profit ratio performance of P-IDA averaged over the 12 groups with $\epsilon = 0.2$ and $b_{min} = 0.1$.
It reveals that the higher value of $a_{max}$ allows achieving higher profit ratio. For example, with $a_{max} =7$, P-IDA achieves on average 70\% profit ratio for these three bidding strategies and decreases dramatically to less than $45\%$ with $a_{max} =5$ and around $13\%$ with $a_{max} =3$. 
This makes sense because for each buyer, the allocation requirement is that the bid price of each buyer should be greater than or equal to the ask price of a feasible seller. When updating the bid prices and ask prices, a large value of $a_{max}$ requires more rounds to terminate the auction, which increases the possibility that a buyer reaches his or her value when the auction terminates. However with small value of $a_{max}$, the bid prices of buyers will quickly catch up with the ask prices of sellers, which increases the chance that P-IDA allocates a buyer with a relatively low bid price. Fig.~\ref{fii3} also shows that the profit ratio achieved by high efficiency bidding strategy, xor-bid-repeating, is slightly higher than xor-bid and single-bid. With $a_{max} = 7$, P-IDA with xor-bid-repeating generates 70\% profit ratio, compared to 63\% for xor-bid and 60\% for single bid.
This indicates that there is a positive correlation between the profit ratio and the efficiency. The high level of profit ratio is accompanied by a high efficiency bidding strategy (xor-bid-repeating). Clearly the cost is the high running time.
\subsection{The performance of P-IDA-SA}
The SA meta-heuristic algorithm is designed to solve the winner determination model in large scale settings where the optimal solutions  cannot be computed in a reasonable time. In this subsection, we test the performance of P-IDA-SA on Groups 13-15. The solutions generated by P-IDA-SA are compared against those generated by the FCFS and also a \textit{greedy allocation mechanism} adapted from~\cite{chichin2016towards}. The greedy allocation mechanism  performs well for resource allocation in two-sided markets, but this is not a truthful mechanism for sellers and also it requires complete valuations of buyers and sellers~\cite{chichin2016towards}.

\begin{table}[h]
\caption{The social welfares of P-IDA-SA, Greedy allocation and FCFS on Groups 13-15.} 
\centering  
\resizebox{0.49\textwidth}{!}{% 
\begin{tabular}{p{1cm}p{1.5cm}p{1cm}p{1cm}p{1.8cm}}
\toprule
Group &P-IDA-SA &Greedy &FCFS & P-IDA-SA Run Time (s)\\
\midrule
13 & 333.8 & 333.5 & 192.1 &5\\
14 & 622.0 & 607.8 &367.2 &49\\
15 & 812.4 & 772.9 & 496.5&310\\
\bottomrule
\end{tabular} } 
\label{table1}
\end{table}

The solutions computed by P-IDA-SA are compared against the greedy allocation mechanism and also FCFS. The second column of TABLE~\ref{table1} shows the average social welfare of each group computed by P-IDA-SA. All buyers are assumed to adopt xor-bid bidding strategy with $\epsilon = 0.2$, $a_{max} = 7$ and $b_{min} = 0.1$. By applying SA to solve the winner determination model, we set the iteration number $R=1000$ and the permutation number $m = 32$. It is observed that the social welfare obtained by the proposed P-IDA-SA is on average 4\% higher than that generated by the greedy allocation mechanism.  The fourth column displays the social welfare computed by FCFS allocation policy. Clearly, our auction achieves a significant improvement over FCFS for all testing groups, achieving on average 40\% higher social welfare. These indicate that P-IDA-SA performs well in realistic large scale settings even without having access to complete information of buyers and sellers. Additional, based on the run time results, the proposed P-IDA-SA is capable of solving large problem instances with high responsiveness. 
\section{Conclusion}
\label{c}
We propose a price-based iterative double auction to compute social welfare maximizing schedules in charger sharing markets.
The proposed auction is suitable for the two-sided structure of the market and possesses desirable economic properties such as budget balance, individual rationality and truthfulness in reporting scheduling constraints. From the theory perspective, we advance the existing literature by extending one-sided auction-based scheduling to two-sided markets by proposing a double auction-based decentralized scheduling mechanism. In addition, the proposed auction is also interesting from practical application perspective. It achieves much better allocative efficiency than the first-come, first-served charger scheduling scheme which has been commonly used by charger sharing platforms.
It also scales well to larger problem  instances, which indicates its potential to be used for large scale charger sharing platforms. In this paper, we only applied the double auction to a day-ahead charging setting. In our future work, we plan to extend the iterative bidding framework to accommodate online dynamic charger sharing scheduling settings. The results presented in this paper pave the foundation and serve as the baseline for the planned extension.

%\bibliography{doubleauction}
\bibliography{Double}
\bibliographystyle{IEEEtran}
\end{document}